\documentclass[a4paper,aps,pra,10pt,twocolumn,notitlepage,superscriptaddress,dvipsnames,tightenlines,floatfix,showpacs, accepted=2020-05-19]{quantumarticle}
\pdfoutput=1
\usepackage{amsmath,amssymb}
\usepackage{xcolor}
\usepackage{bm}
\usepackage{bbold}
\usepackage{mathtools}
\usepackage{amsthm}
\usepackage{graphicx}
\usepackage{paralist}
\usepackage{qcircuit}
\usepackage[normalem]{ulem}
\usepackage{enumitem}
\setlist{leftmargin=.8cm}

\usepackage{xcolor}
\definecolor{dark-red}{rgb}{0.4,0.15,0.15}
\definecolor{dark-blue}{rgb}{0.15,0.15,0.4}
\definecolor{medium-blue}{rgb}{0,0,0.5}
\usepackage{hyperref}
\usepackage[all]{hypcap}
\hypersetup{
    colorlinks, linkcolor={dark-red},
    citecolor={dark-blue}, urlcolor={medium-blue}
} 

\newcommand{\nsp}{\hspace{-0.4pt}}
\newcommand{\ssp}{\hspace{0.4pt}}
\newcommand{\norm}[1]{\lvert #1 \rvert}

\newcommand{\ket}[1]{\lvert\ssp #1\ssp \rangle}

\newcommand{\bra}[1]{\langle\ssp #1\ssp \rvert}

\newcommand{\braket}[2]{\langle\, #1\,\vert\, #2 \,\rangle}

\newcommand{\anticommute}[2]{{\{\ssp #1,\, #2 \ssp\}}}

\newcommand{\DD}{D}
\newcommand{\cc}{c}
\newcommand{\dd}{d}

\newcommand{\hh}{h}
\newcommand{\pp}{p}
\newcommand{\bound}{2n}
\DeclareMathOperator*{\Motimes}{\text{\raisebox{0.25ex}{\scalebox{0.8}{$\bigotimes$}}}}

\DeclareMathOperator{\tr}{Tr}

\newtheorem{theorem}{Theorem} 

\graphicspath{{./Figures/}}

\begin{document}

\title{Optimal fermion-to-qubit mapping via ternary trees with applications to reduced quantum states learning}
\author{Zhang Jiang}
\email{qzj@google.com}
\affiliation{Google Research, Venice, CA 90291}
\author{Amir Kalev}
\email{amirk@umd.edu}
\affiliation{Joint Center for Quantum Information and Computer Science,
University of Maryland, College Park, MD 20742-2420, USA}
\author{Wojciech Mruczkiewicz}
\affiliation{Google Research, Venice, CA 90291}
\author{Hartmut Neven}
\affiliation{Google Research, Venice, CA 90291}
\date{2020-05-26}
\pacs{03.65.Ud, 03.67.Lx, 06.20.Dk}

\begin{abstract}
We introduce a fermion-to-qubit mapping defined on ternary trees, where any single Majorana operator on an $n$-mode fermionic system is mapped to a multi-qubit Pauli operator acting nontrivially on $\lceil \log_3(2n+1)\rceil$ qubits.  The mapping has a simple structure and is optimal in the sense that it is impossible to construct Pauli operators in any fermion-to-qubit mapping acting nontrivially on less than $\log_3(\bound)$ qubits on average.  We apply it to the problem of learning $k$-fermion reduced density matrix (RDM), a problem relevant in various quantum simulation applications. We show that one can determine individual elements of all $k$-fermion RDMs in parallel, to precision $\epsilon$, by repeating a single quantum circuit for $\lesssim (2n+1)^k  \epsilon^{-2}$ times.  This result is based on a method we develop here that allows one to determine individual elements of all $k$-qubit RDMs in parallel, to precision $\epsilon$, by repeating a single quantum circuit for $\lesssim 3^k \epsilon^{-2}$ times, independent of the system size.  This improves over existing schemes for determining qubit RDMs.

\end{abstract}

\maketitle

\section{Introduction}
\label{sec:intro}

Quantum simulation opens an alternative approach to solving hard problems ubiquitous in physics, chemistry, and material sciences (e.g., high temperature superconductivity)~\cite{feynman_simulating_1982,lloyd_universal_1996,georgescu_quantum_2014,wecker_solving_2015,babbush_low-depth_2018,jiang_quantum_2018}.  With the rapid advances in quantum computing devices, such as trapped ions~\cite{cirac_quantum_1995,kielpinski_architecture_2002,haffner_quantum_2008} and superconducting qubits~\cite{devoret_superconducting_2004,wendin_quantum_2017}, we are closer today to realize this goal.  An essential step in digital quantum simulation of a fermionic systems is mapping it onto a qubit system.  Having an efficient, simple, fermion-to-qubit mapping is a key ingredient in any quantum simulation protocol, e.g., the variational quantum eigensolver (VQE)~\cite{peruzzo_variational_2014,mcclean_theory_2016}.       

The most common fermion-to-qubit mapping is the Jordan-Wigner transformation (JWT)~\cite{nielsen_fermionic_2005}, which maps single fermionic operators on an $n$-mode fermionic system to qubit operators acting nontrivially on $\mathcal O(n)$ qubits.  In the JWT, the qubit state $\ket{0}$ ($\ket{1}$) represents an occupied (unoccupied) fermionic mode, i.e., its occupation information.  Instead, one can also store the parity information of the fermionic system~\cite{seeley_bravyi-kitaev_2012}, where the $j$-th qubit state is $\ket{0}$ ($\ket{1}$) if the total number of fermions in the first $j$ fermionic modes is even (odd).  By balancing the locality of occupation and parity information, the Bravyi-Kitaev transformation (BKT)~\cite{bravyi_fermionic_2002} maps single fermionic operators to qubit operators acting nontrivially on $\mathcal O(\log_2 n)$ qubits.  The BKT is important because it avoids operations acting on an extensive number of qubits.  It is also quite involved, and many authors have discussed its applications in quantum simulation of fermionic systems~\cite{seeley_bravyi-kitaev_2012, tranter_bravyikitaev_2015}.  Recently~\cite{havlicek_operator_2017}, the BKT was reformulated using the Fenwick trees, a classical data structure that allows for efficient updating elements and calculating prefix sums. 

As the possibilities to study fermionic systems on quantum computing devices materialize~\cite{arute_quantum_2019}, it is well timed to explore other, possibly simpler and more efficient fermion-to-qubit transformations, beyond existing ones.  Here, we present such a transformation.  We construct a simple fermion-to-qubit mapping defined on ternary trees.  It maps any single Majorana operator on an $n$-mode fermionic system to a multi-qubit Pauli operator acting nontrivially on $\lceil \log_3(2n+1)\rceil$ qubits.  For large $n$, it is approximately $\log_2 3 \simeq 1.58$ times lower than $\lceil\log_2 n+1\rceil$ in the BKT~\cite{havlicek_operator_2017}.  We prove that the operator weight in the ternary-tree mapping is optimal.  The ternary-tree mapping was also introduced to derive representations of the Clifford algebras and spin group~\cite{vlasov_clifford_2019}. 

In quantum simulation, it is often desirable to learn the quantum state involving only a fixed numbers of qubits (fermions), i.e., reduced density matrices (RDMs) tomography~\cite{nielsen_quantum_2000, sharma_reduced_2008, fagotti_reduced_2013, rubin_application_2018}.  Learning the $k$-particle RDM not only allows one to estimate the ground state energy of a fermionic system with up to $k$-particle interactions, it also enables one to derive a large number of important properties of the system, such as multipole moments~\cite{gidofalvi_molecular_2007} and derivatives of energy~\cite{overy_unbiased_2014,obrien_calculating_2019}; moreover, it is an indispensable part in error-mitigating techniques by relaxing the single-particle orbitals~\cite{mcclean_hybrid_2017,takeshita_increasing_2020}.  Using the ternary-tree mapping, we show that individual elements of the $k$-particle RDM of an $n$-mode fermionic state stored in a quantum computer can be determined in parallel to precision $\epsilon$ by repeating a single quantum circuit for $\lesssim (2n+1)^k / \epsilon^2$ times.  

The technique to obtain the fermionic RDMs is based on a method we develop here that allows one to determine individual elements of all $k$-qubit RDMs of an $n$-qubit system in parallel, to a precision $\epsilon$, by repeating a single quantum circuit for $\lesssim 3^k / \epsilon^2$ times, independent of the system size $n$.  Our scheme is based on measuring in the Bell basis of a system qubit and an ancilla qubit.  It improves the scaling of a recent result~\cite{cotler_quantum_2020,bonet-monroig_nearly_2019} by a $\log n$ factor.  Moreover, it is informationally complete, implying that one can retrieve the entire quantum state by simply repeating the same quantum circuit many times.   

The paper is organized into two parts.  In Sec.~\ref{sec:mapping}, we introduce the ternary-tree mapping and prove its optimality.  In Sec.~\ref{sec:RDM}, we discuss the Bell-basis measurement scheme for quantum tomography of qubits, and apply it to fermions using the ternary-tree mapping.  In App.~\ref{sec:sic_qubit}, we prove that the Bell-basis measurement scheme for qubits implements a symmetric informationally complete (SIC) POVM.  In App.~\ref{sec:sic_qudits}, we generalize the Bell-basis measurement scheme to qudits and prove its informationally completeness and discuss its relation to SIC POVMs.

\section{Fermion-to-qubit mapping}
\label{sec:mapping}

In the second quantization formulation to quantum mechanics, the $n$-mode fermionic operators can be expanded using the Majorana fermion operators
\begin{align}
  \gamma_{2j} = \cc^\dagger_j + \cc_j\,,\quad \gamma_{2j+1} = i\ssp \big(\cc^\dagger_j - \cc_j \big)\,,
\end{align}
for $j = 1, 2, \ldots, n$, where $\cc_j$ and $\cc^\dagger_j$ are the $j$-th annihilation and creation operators.  The Majorana operators satisfy the simple anticommutation relation $\anticommute{\gamma_u}{\gamma_v} = 2\delta_{uv}$ for $u,\, v = 1,\ldots, 2n$.  In what follows we construct $2n$ independent multi-qubit Pauli operators that are mutually anticommute, which can be used to represent the Majorana operators. 

Our fermion-to-qubit mapping is defined on a ternary tree, where each node (except those in the bottom level of the tree) is associated with a qubit.  We start with the case that the tree is complete and come back to the more general cases later on.  The total number of qubits in a complete ternary tree of height $h$ is
\begin{align}
n = \sum_{\ell=0}^{\hh - 1}\, 3^{\ell} = \frac{3^{\hh} -1}{2}\,.
\end{align}
We label the qubit associated with the root node by $0$, the rest of the qubits are indexed consecutively, as we are going down the tree, see Fig.~\ref{fig:tree}.  The Pauli operators of the $\eta$-th qubit are denoted as $\sigma^{x,\,y,\,z}_{\eta}$.  Each root-to-leaf path on the tree can be uniquely specified by the vector $\pmb \pp = (\pp_0, \ldots, \pp_{\hh-1})$, where $\pp_\ell=0,\,1,\,2$ determines the next node on the path with depth $\ell + 1$ based on the current node with depth $\ell$.  We can write the index of a node of depth $\ell$ on the path $\pmb p$ by
\begin{align}
 \eta(\ssp\ssp\pmb \pp,\ssp\ell) = \frac{3^{\ell} -1}{2} + \sum_{j=0}^{\ell-1}\, 3^{\ell-1-j}\,\pp_j\,,\quad\text{for $\ell \geq 0$}\,,
\end{align}
where $(3^{\ell} -1)/2$ is the number of nodes with depth less than $\ell$.  Next, we associate a multi-Pauli operator, $A_{\pmb \pp}$, with each root-to-leaf path $\pmb p$,
\begin{align}
A_{\pmb \pp} = \prod_{\ell=0}^{\hh-1} \sigma^{\chi(\pp_\ell)}_{\eta(\ssp\ssp\pmb \pp,\ssp\ell)}\,,\quad A_{\pmb \pp}^2=\openone\,,
\end{align}
where $\chi(\pp) = x,\,y,\,z$ for $\pp=0,\,1,\,2$, respectively.   By construction these operators mutually anticommute, i.e., $\anticommute{A_{\pmb p}}{A_{\pmb q}} = 0$ for $\pmb p\neq \pmb q$.  This is because $\pmb p$ and $\pmb q$ have to diverge at some point.  Before that point they involve the same Pauli operators on the same qubits, at the point of divergence $A_{\pmb p}$ and $A_{\pmb q}$ associate different Pauli operator to the same qubit, whereas after that point they involve Pauli operators on different qubits.

\begin{figure}[!tb]
\centering
\includegraphics[width=0.95\columnwidth]{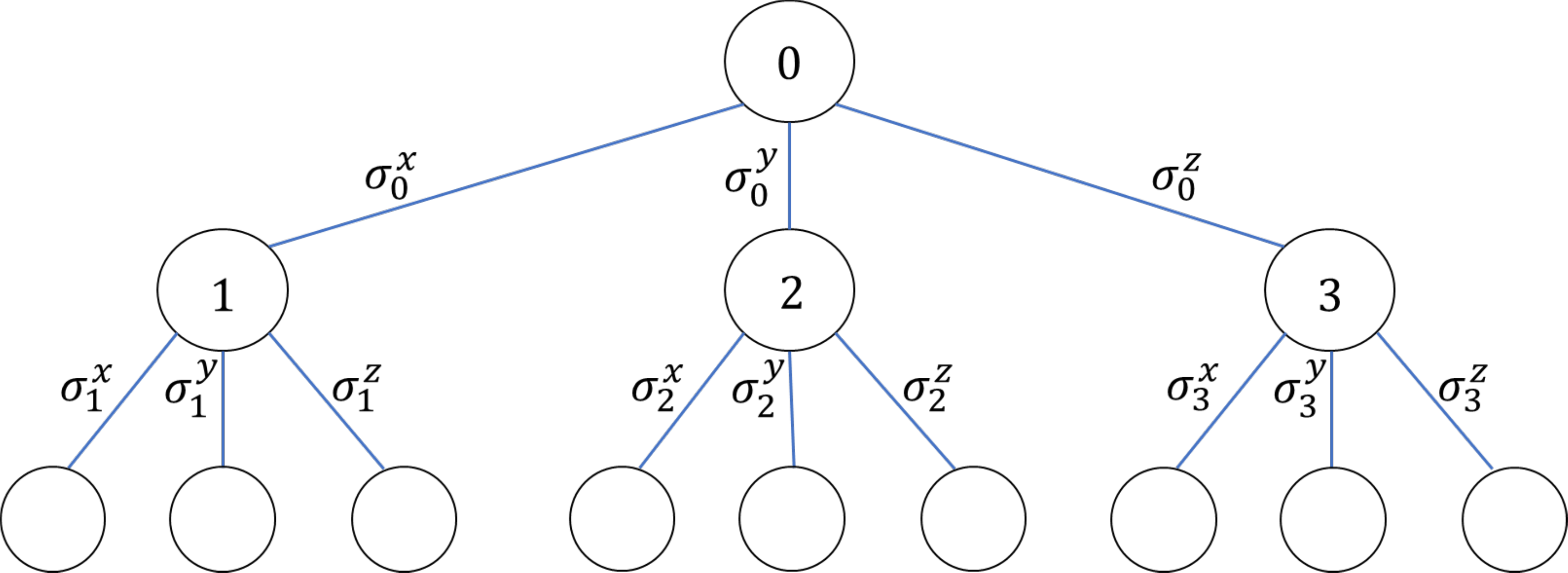}
\caption{{\bf Ternary-tree mapping.} An example of the fermion-to-qubit mapping in a ternary tree of height $2$.  The label of the qubit associated with a node is written inside it.  We introduce a Pauli operator for each root-to-leaf path in the tree.  For example, the left-most path corresponds to the Pauli operator $A_{0,0}=\sigma_0^x\sigma_1^x$.}\label{fig:tree}
\end{figure}

There are $3^\hh = 2n+1$ distinct root-to-leaf paths in the ternary tree, whereas the total number of independent operators $A_{\pmb p}$ is $2n$.  This is because the product of $A_{\pmb p}$ for all paths ${\pmb p}$ is proportional to the identity operator.  Therefore, we can map $2n$ Majorana operators to $2n$ independent Pauli operators that mutually anticommute.  The Pauli weight of $A_{\pmb p}$ equals to the tree height $\hh = \log_3 (2n+1)$.  In the large $n$ limit, it is approximately $\log_2 3 \simeq 1.58$ times lower than $\lceil\log_2 n+1\rceil$ in the BKT~\cite{havlicek_operator_2017}.
This reduction is achieved by balancing all of the three Pauli operators, whereas only Pauli-$x$ and $z$ are balanced in the BKT.  

Consider the case that $2n+1$ is not a power of 3, i.e., $3^{\hh} < 2n+1 < 3^{\hh+1}$ for some nonnegative integer $\hh$.  We can still construct a ternary-tree mapping with Pauli weights less or equal to $\lceil\log_3 (2n+1)\rceil$  by taking the following steps:
\begin{enumerate}
\item{Introduce a complete ternary tree with height $\hh$,}
\item{Choose $n- (3^{\hh}-1)/2$ } leaves in the tree and associate a qubit with each of them,
\item{Branch the corresponding $A_{\pmb \pp}$ to $A_{\pmb \pp}\otimes \sigma_{\pmb \pp}^{x, y, z}$, where $\sigma_{\pmb \pp}^{x, y, z}$ are the Pauli operators on the added qubit,}
\item{Map the resulting Pauli operators, of weights $\hh$ and $\hh+1$, to Majorana operators.}
\end{enumerate}
We now prove that the ternary-tree mapping is optimal. 
\begin{theorem}\label{thm:average_wieght}
The averaged weight of the Pauli operators in any map that transforms an $n$-mode fermionic system onto a qubit system is at least $\log_3 (\bound)$.
\end{theorem}
\begin{proof} 
We introduce the real vector ${\mathbf r}^\top = ( r_1, \ldots, r_{2n})$, where $r_u =  \langle\gamma_u\rangle=\langle\gamma_u^\dagger\rangle$ is the expectation value of the $u$-th Majorana operator. 
We also introduce the $2n\times 2n$ Hermitian matrix $G$ whose elements are given by
\begin{align}
G_{uv} = \big\langle
(\gamma_u - r_u)
(\gamma_v - r_v) \big\rangle=
\langle\gamma_u\gamma_v\rangle - r_u\ssp r_v\,.
\end{align} 
Since $G$ is a Gram matrix it is positive semidefinite.  Consider the expectation value
\begin{align}
{\mathbf r}^\top\! G\, \mathbf r &= \sum_{u,\,v} r_u\ssp r_v\, \langle\gamma_u\gamma_v\rangle - r_u^2\, r_v^2 \nonumber\\
&= {\mathbf r}^\top {\mathbf r} - ({\mathbf r}^\top \mathbf r)^2 \geq 0\,,
\end{align}
where we use $\langle\gamma_u\gamma_v\rangle = -\langle\gamma_v\gamma_u\rangle$ for $u\neq v$ and $\langle\gamma_u\gamma_u\rangle = 1$;  therefore, we have ${\mathbf r}^\top \mathbf r \leq 1$ for any fermionic state.   We then consider an $m$-qubit product state $\xi^{\otimes m}$, where $\xi$ is the single-qubit pure state
\begin{align}\label{eq:equal_distance_state}
 \xi=\frac{1}{2}\bigg(\openone+\frac1{\sqrt3}\big(\sigma^x+\sigma^y+\sigma^z\big)\bigg)\,,
\end{align}
which leads to $\tr(\sigma^{x}\xi)=\tr(\sigma^{y}\xi)=\tr(\sigma^{z}\xi)=1/{\sqrt3}$.   This particular pure product state is symmetric with respect to $\sigma_x$,  $\sigma_y$ and  $\sigma_z$; therefore, it allows us to treat expectation values of Pauli operators with the same weight on an equal footing.   If  $\gamma_u$ is mapped to a Pauli operator with weight $w_u$, we have $\norm{r_u} = 1/\sqrt{3}^{\,w_u}$ under the product state $\xi^{\otimes m}$, and thus
\begin{align}\label{eq:inequality}
   2n\, 3^{-\overline w} \leq  \sum_{u=1}^{2n} 3^{-w_u} =  {\mathbf r}^\top \mathbf r\leq 1\,,
\end{align}
where $\overline w = \frac{1}{2n}\,\sum_{u=1}^{2n} w_u$ and we use the fact that $y=3^{-x}$ is a convex function in the leftmost inequality.   As a consequence, we have
\begin{align}
   \overline w \geq \log_3 (2n)\,.
\end{align}
When the Pauli weight $w_u = w$ is a constant, we have
\begin{align}
   w \geq \lceil\log_3 (2n)\rceil\geq \log_3 (2n+1)\,.
\end{align}
Therefore, the ternary-tree mapping that we have constructed is optimal.
\end{proof}

\section{Reduced density matrices tomography}
\label{sec:RDM}

\subsection{Qubits}

In this section, we first introduce a method of reconstructing RDMs specifically for qubit systems, and then use this method together with the ternary-tree mapping to reconstruct RDMs for fermionic systems.

Given an $n$-qubit quantum state $\rho$, the $k$-qubit reduced density matrix ($k$-RDM) may be written as
\begin{align}\label{eq:krdm_qubit}
   \tr_{\substack{\neq j_1,\ldots, j_k}}\big( \rho\big) = \frac{1}{2^k}\!\sum_{\substack{\alpha_1,\ldots,\alpha_k\, =\, 0,\, x,\, y,\, z}}\!\! \rho^{\alpha_1,\ldots,\alpha_k}_{j_1,\ldots, j_k}\Motimes_{i=1}^k \sigma^{\alpha_i}_{j_i}\,,
\end{align}
 where $\sigma_j^{\alpha}$ is the Pauli operator on the $j$-th qubit, $\sigma^{0} = \openone$ and $\sigma^{x,\,y,\,z}$ are the Pauli-$x$, -$y$, and -$z$ operators, respectively.  The matrix elements of the $k$-RDMs in Eq.~\eqref{eq:krdm_qubit} are defined as 
\begin{align}\label{eq:kcorr}
\rho^{\alpha_1,\ldots,\alpha_k}_{j_1,\ldots, j_k} = \tr\bigg(\rho\,\Motimes_{i=1}^k \sigma_{j_i}^{\alpha_i}\bigg)\,.
\end{align}
Assuming all $(k-1)$-RDMs are known, measuring the $\binom{n}{k}\ssp 3^k$ different observables, $\Motimes_{i=1}^k \sigma_{j_i}^{\alpha_i}$ ($\alpha_i\neq 0$), provides us with the required information to reconstruct all $k$-RDMs.  Under the assumption that only single-qubit operations are allowed, Cotler and Wilczek~\cite{cotler_quantum_2020} showed that elements in all $k$-RDMs can be sampled by implementing $e^{\mathcal O(k)}\log(n)$ different quantum circuits.  

Changing quantum circuits is typically quite slow on current programmable quantum devices, e.g., those based on FPGAs, while repeating a single circuit may be done much faster.  To circumvent the problem of programming various circuits, we propose a scheme that allows one to estimate all $k$-qubit RDMs at once.  Our approach is based on measuring a system qubit and an ancilla qubit in the Bell basis
\begin{align}
 &\ket{\Phi^\pm} = \frac{1}{\sqrt{2}} \Big(\ket{0} \otimes \ket{0} \pm \ket{1}\otimes \ket{1}\Big)\,,\\[3pt]
&\ket{\Psi^\pm} = \frac{1}{\sqrt{2}} \Big(\ket{0} \otimes \ket{1} \pm \ket{1}\otimes \ket{0}\Big)\,.
\end{align}
The Bell basis is a common eigenbasis of the commuting operators $\sigma^x\otimes\sigma^x$, $\sigma^y\otimes\sigma^y$, and $\sigma^z\otimes\sigma^z$, and their eigenvalues are listed in Tab.~\ref{fig:bell_basis}.  
\begin{table}[htb]
    \centering
    {\renewcommand{\arraystretch}{1.5}
\begin{tabular}{|c||c|c|c| }
 \hline
 \;Outcome\;& \;$\sigma^x\otimes \sigma^x$\; &\;$\sigma^y\otimes \sigma^y$\; & \;$\sigma^z\otimes \sigma^z$\;\\
 \hline
 $\Phi^+$ & 1 & -1 & 1\\
 $\Phi^-$ &-1 &  1 & 1\\
 $\Psi^+$ & 1 &  1 & -1\\
 $\Psi^-$ &-1 & -1 & -1\\
 \hline
\end{tabular}}
    \caption{Eigenvalues of the Pauli tensor products in the Bell basis. }
    \label{fig:bell_basis}
\end{table}
Therefore, the expectation values of these three operators can be measured simultaneously.  The Bell measurements can be done in parallel for all pairs of qubits using Hardmard and CNOT gates.  The quantum circuit to implement our scheme, for one system qubit and one ancilla qubit, is plotted in Fig.~\ref{fig:ibmcircuit}.  

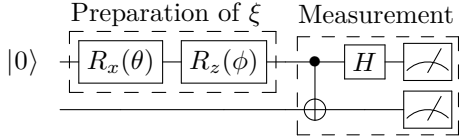
\begin{figure}[hbt]\label{fig:mu_circuit}
\hspace{1em}\Qcircuit @C=0.7em @R=0.2em @!R { 
& &  &  \hspace{1.2cm}\mbox{Preparation of $\xi$} & & & & \hspace{0.3cm}\mbox{Measurement} \\
\ket{0} & & &\gate{R_x(\theta)}&\gate{R_z(\phi)}&\qw &\ctrl{1}&\gate{H}&\meter \\
 & & & \qw&\qw&\qw&\targ&\qw&\meter \gategroup{2}{4}{2}{5}{.7em}{--}\gategroup{2}{7}{3}{9}{.5em}{--}
}
\caption{A circuit description of the proposed Bell-basis measurement on a single system qubit, represented by the lower wire. The top wire represents an ancilla qubit, where the rotation angles are $\theta=\arccos(\frac1{\sqrt{3}})$ and $\phi=\frac{3\pi}{4}$.}
\label{fig:ibmcircuit}
\end{figure}

If the Pauli-$x$, -$y$, and -$z$ operators are sampled at the same rate, a natural choice of the ancilla state is the pure state defined in Eq.~\eqref{eq:equal_distance_state}.  It can be obtained by rotating $\ket{0}$ about the $x$ axis by an angle $\theta=\arccos(1/\sqrt{3}\,)$ followed by a rotation around the $z$ axis by an angle $\phi=3\pi/4$, see Fig.~\ref{fig:mu_circuit}.  Measuring the system and ancilla qubits in the Bell basis yields
\begin{align}
\tr\bigg(\rho\otimes\xi^{\otimes n} \prod_{i=1}^k \sigma^{\alpha_i}_{j_i}\otimes \sigma^{\alpha_i}_{j'_i}\bigg)=\frac1{\sqrt{3}^{\,k}}\;\rho^{\alpha_1,\ldots,\alpha_k}_{j_1,\ldots, j_k}\,,
\end{align}
where we use $\tr(\xi\ssp \sigma^x) = \tr( \xi\ssp \sigma^y) =\tr( \xi\ssp \sigma^z) = 1/\sqrt 3$.
Due to the factor $1/\sqrt{3}^{\, k}$, we must run the experiment $3^k /\epsilon^2$ times to obtain the standard-deviation error $\epsilon$.  

In summary, our scheme measures all $k$-qubit RMD elements in parallel, with error that scales as $\epsilon$, by running a single quantum circuit for $3^k / \epsilon^2$ times:
\begin{enumerate}
\item{To each system qubit (labeled by $j$) attach an ancillary qubit (labelled by $j'$) in a known state $\xi$, so that the total system-ancilla state is $\rho\otimes\xi^{\otimes n}$.}
\item{Measure each pair of qubits $(j,j')$ in the common eigenbasis of $\sigma^x\otimes\sigma^x$, $\sigma^y\otimes\sigma^y$, and $\sigma^z\otimes\sigma^z$, i.e., the Bell basis.}
\end{enumerate}
While the scheme in~\cite{cotler_quantum_2020} allows one to specify all RDMs up to some order $k$, our scheme allows for specifying the entire $n$-qubit state, i.e., all RDMs, for all $k\in[1,n]$ ~\footnote{For finite rounds of measurements, though we have information about all RDMs, the ones involving fewer qubits (smaller $k$) would be estimated more accurately than those involving more qubits (lager $k$).}.  The reason for this inherent difference is rooted in that the measurement we perform  on each qubit is {\it informationally complete}, that is it allows us to estimate the expectation values of $\sigma_x$, $\sigma_y$, and $\sigma_z$, from one POVM.  Since this measurement is done in parallel for all qubits, we can use the data to estimate the entire $n$-qubit state.  In Apps.~\ref{sec:sic_qubit}, we show it implements a symmetric informationally-complete measurement (SIC POVM) on a qubit, a.k.a., the tetrahedron measurement~\cite{Englert04}. Finally, we note that recent work~\cite{hamamura_efficient_2019} has also suggested to use Bell-measurements in the context of RDMs reconstruction. However, there the use of Bell-measurements was done on pairs of system qubits.  In App.~\ref{sec:sic_qudits}, we generalize the Bell-basis measurement method to $D$-dimensional system, and show that similar to the qubit case it is informationally-complete and implements SIC POVMs in certain cases.

\subsection{Fermions}

Now, we apply the Bell-basis measurement method for RDM reconstruction together with the ternary-tree mapping for reconstructing $k$-fermion RDMs.  Generally, it requires $\mathcal O(n^{2k})$ parameters to determine the $k$-RDMs of an $n$-mode fermionic system.  To make matters worse, fermions obey anticommutation relations which prevents one from finding large groups of commuting operators that can be measured simultaneously.  In second quantization, the elements in a fermionic $k$-RDM of $n$ modes can be expressed as expectation values involving $2k$ Majorana operators.  For instance, the fermionic 2-RDM may be written as
\begin{align}
 \rho_{stuv} =  \big\langle \gamma_s\gamma_t\gamma_u\gamma_v\big\rangle  \,, 
\end{align}
where the $\gamma$'s are the Majorana fermion operators defined above.  Encoded with the ternary-tree mapping, the Pauli weights of elements in fermionic $k$-RDMs are at most $2k\log_3 (2n+1)$.  With the Bell-basis measurement strategy, the matrix elements are attenuated by a factor bounded by 
\begin{align}\label{eq:fermi_rdm}
    {\sqrt 3}^{\: 2k\log_3 (2n+1)} = (2n+1)^k\,.
\end{align}
Therefore, all elements in a fermionic $k$-RDMs can be measured to precision $\epsilon$ by repeating the same circuit for $\lesssim(2n+1)^k / \epsilon^2$ times.

Recently, Bonet-Monroig, Babbush, and O'Brien~\cite{bonet-monroig_nearly_2019} showed that if one is able to implement a linear depth circuit on a linear array prior to measurement, then one can directly measure the fermionic $2$-RDM using $\mathcal O(n^2)$ circuits.  This result, albeit different from ours, suggests the same scaling of measurement repetitions to sample elements in fermionic RDMs on a quantum computer.  If only the ground state energy is concerned, more efficient methods based on linear combination of unitary operators~\cite{izmaylov_unitary_2020, zhao_measurement_2019} and a low rank factorization of the two-electron integral tensor~\cite{huggins_efficient_2019} are recently proposed.

\section{Conclusion}
\label{sec:con}

The fermionic anticommutation relation has many important implications, including the Pauli exclusion principle and more generally the $n$-representability problem~\cite{mazziotti_reduced-density-matrix_2009}.  As a direct consequence of the anticommutation relation, we prove a lower bound on the weights of Pauli operators to represent single fermionic operators of an $n$-mode fermionic system.  We also constructed a fermion-to-qubit mapping based on ternary trees that saturates this bound, where the numbers of Pauli-$x$, -$y$, and -$z$ operators are balanced.  We expect the ternary-tree mapping to be a useful tool to studying fermionic systems, such as Hubbard-like systems, on noisy and fault-tolerant quantum computing devices.


Related to the fermion-to-qubit mapping, we discuss the measurement of fermionic RDMs which is likely to be a bottleneck in quantum simulation of fermions; the noncommuting feature of fermionic operators forbids one to measure them simultaneously.  The fermionic anticommutation relation again puts a lower bound on the measurement repetitions required to determine the RDMs.  Here, we have shown that individual elements of the $k$-RDM of a fermionic state---stored in a quantum computer with the ternary-tree mapping---can be determined in parallel to precision $\epsilon$ by repeating a single quantum circuit for $\lesssim (2n+1)^k / \epsilon^2$ times.  In comparison, using the Bravyi-Kitaev transformation leads to a worse exponent in the number of repetitions, i.e., $\lesssim (n+1)^{k \log_2\nsp 3} / \epsilon^2$.

These measurement results for fermions are based on a Bell-basis method we develop here that allows for determining individual elements of all $k$-qubit RDMs in parallel, to precision $\epsilon$, by repeating a single quantum circuit for $\lesssim 3^k \epsilon^{-2}$ times, independent of the system size.  This is especially suited for quantum computers with slow circuit updating rates, e.g., those based on FPGAs.  Our scheme also saves a $\log n$ factor in circuit repetitions compared to the results by Cotler and Wilczek~\cite{cotler_quantum_2020}, at the cost of introducing extra ancilla qubits.  Another desired property of our scheme is that it is informationally complete, allowing for reconstruction of the entire quantum state as opposed to the RDMs up to some order $k$.  We have also generalized the Bell-basis measurement scheme to qudits using the Heisenberg-Weyl operators.

Future work includes designing fermion-to-qubit mappings on ternary trees with varying depths of root-to-leaf paths.  This is useful to problems where some of the Majorana operators are more important and need to be sampled more often.  It is desirable to map these operators to Pauli operators with less weights, corresponding to low-depth paths in the tree.  Another interesting line of research is to combine the ternary-tree mapping with quantum error-correcting codes by introducing redundant qubits~\cite{steudtner_quantum_2019, setia_bravyi-kitaev_2018, jiang_majorana_2019}.

\begin{acknowledgments}
After completion of this work, we became aware of the work by Vlasov~\cite{vlasov_clifford_2019}, and we thank the author for bringing his work to our attention.  We would like to thank Sergio Boixo, Yu-An Chen, Sam McArdle, Jarrod McClean, Tom O'Brien, Nick Rubin, Kevin Sung, and Vadim Smelyanskiy for fruitful discussions.  Special thanks to Ryan Babbush for describing the problem and the beautiful results he and collaborators found during a surfing trip.  AK acknowledges the support from the US Department of Defense and the support from the AFOSR MURI project ``Scalable Certification of Quantum Computing Devices and Networks.''
\end{acknowledgments}


%

\onecolumngrid
\appendix

\section{The informational completeness of the Bell-basis measurement and its implementation}
\label{sec:sic_qubit} 

In this Appendix, we prove that the Bell-basis measurement scheme introduced in Sec.~\ref{sec:RDM} implements a symmetric informationally complete (SIC) POVM for qubits.

\begin{theorem}\label{thm:qubit_info_comp}
If $\tr(\sigma^\alpha\xi)\neq0$ for $\alpha=x,y,z$, our procedure is informationally complete for $\rho$ in the limit of infinite number of measurement repetitions, i.e., $\rho$ can be constructed from the measurement outcomes.
\end{theorem}
\begin{proof}
The Bell basis
\begin{align}
 &\ket{\Phi^\pm} = \frac{1}{\sqrt{2}} \Big(\ket{0} \otimes \ket{0} \pm \ket{1}\otimes \ket{1}\Big)\,,\\[3pt]
&\ket{\Psi^\pm} = \frac{1}{\sqrt{2}} \Big(\ket{0} \otimes \ket{1} \pm \ket{1}\otimes \ket{0}\Big)\,.
\end{align}
is a common eigenbasis of $\sigma^x\otimes\sigma^x$, $\sigma^y\otimes\sigma^y$, and $\sigma^z\otimes\sigma^z$.  Therefore, in the limit of infinite number of measurement repetitions, for any $k\in[1,n]$,  $\alpha_i=x,y,z$ and $j_i\in[1,n]$ ($j_i\neq j_{i'}$ for $i\neq i'$, $i, i'=1,\ldots,k$), using the procedure described in the main text we can calculate the $2k$-body correlation function $\tr(\sigma^{\alpha_1}_{j_1}\otimes \sigma^{\alpha_1}_{j'_1} \cdots\sigma^{\alpha_k}_{j_k}\otimes \sigma^{\alpha_k}_{j'_k}\rho\otimes\xi^{\otimes n})$ which equals to $\tr(\sigma^{\alpha_1}_{j_1} \cdots\sigma^{\alpha_k}_{j_k}\rho)\prod_{i=1}^k\tr(\sigma^{\alpha_i}_{j'_i}\xi)$. Since we assumed $\xi$ is known and $\tr(\sigma^\alpha\xi)\neq0$  for $\alpha=x,y,z$, we can write
\begin{align}
\rho^{\alpha_1,\ldots,\alpha_k}_{j_1,\ldots, j_k} &=\tr(\sigma^{\alpha_1}_{j_1} \cdots\sigma^{\alpha_k}_{j_k}\rho)\\&
=\frac{\tr(\sigma^{\alpha_1}_{j_1}\otimes \sigma^{\alpha_1}_{j'_1} \cdots\sigma^{\alpha_k}_{j_k}\otimes \sigma^{\alpha_k}_{j'_k}\rho\otimes\xi^{\otimes n})}{\prod_{i=1}^k\tr(\sigma^{\alpha_i}_{j'_i}\xi)}\,. 
\end{align}
\end{proof}

The SIC POVM on a qubit is a collection of four outcomes (POVM elements) $E_i=\frac1{2}\ket{\psi_i}\bra{\psi_i}$, $i=1,\ldots,4$, such that 
\begin{align}
\vert\braket{\psi_i}{\psi_j}\vert^2=\frac{2\delta_{i,j}+1}{3}\,. 
\end{align}
We now show that, following Neumark's theorem it can be implemented  by attaching an ancilla qubit in a state $\xi=\frac1{2}(\openone+\frac1{\sqrt3}(\sigma_x+\sigma_y+\sigma_z))$ and measuring the two qubits in the Bell-basis.

In the Bell-basis measurement, the probability to obtain the outcome that corresponds to $\Phi^{+}$ is:
\begin{align}
\bra{\Phi^{+}}\rho\otimes\xi\ket{\Phi^{+}}&=\frac{1}{2}\Big(\bra{0}\rho\ket{0}\bra{0}\xi\ket{0}+\bra{0}\rho\ket{1}\bra{0}\xi\ket{1}+\bra{1}\rho\ket{0}\bra{1}\xi\ket{0}+\bra{1}\rho\ket{1}\bra{1}\xi\ket{1}\Big)\,. 
\end{align}
Therefore the POVM element on the system qubit which corresponds to the $\Phi^{+}$ outcome is
\begin{align}\label{eq:povm1}
E_{\Phi^{+}}&=\frac{1}{2}\Big(\bra{0}\xi\ket{0}\ket{0}\bra{0}+\bra{0}\xi\ket{1}\ket{1}\bra{0}+\bra{0}\xi\ket{1}\ket{1}\bra{0}+\bra{1}\xi\ket{1}\ket{1}\bra{1}\Big)=\frac{\xi}{2}\,. 
\end{align}
Similarly we find,
\begin{align}
E_{\Phi^{-}}&=\frac{1}{2}\sigma^z\xi\ssp\sigma^z\,,\label{eq:povm2}\\ 
E_{\Psi^{+}}&=\frac{1}{2}\sigma^x\xi\ssp\sigma^x\,,\label{eq:povm3}\\ 
E_{\Psi^{-}}&=\frac{1}{2}\sigma^y\xi\ssp\sigma^y\label{eq:povm4}\,.
\end{align}
Note that $\xi$ is a pure state and therefore so are $\sigma^x\xi\sigma^x$, $\sigma^y\xi\sigma^y$, and $\sigma^z\xi\sigma^z$. To see that  $\{E_{\Phi^{+}},E_{\Phi^{-}},E_{\Psi^{+}},E_{\Psi^{-}}\}$ forms a SIC POVM on a qubit, one can verify explicitly that for $\alpha\neq\beta=0,x,y,z$ $\tr(\sigma^\alpha\xi\sigma^\alpha \sigma^\beta\xi\sigma^\beta)=\frac{1}{3}$ where $\sigma^0=\openone$. 

The implementation of the tetrahedron measurement on a qubit using the Bell-basis measurement points at a deep connection between dense coding and SIC POVM of a qubit. More explicitly, the protocol of dense coding is based on the relations
\begin{align}
\ket{\Phi^{+}}&=\openone\otimes\openone\ket{\Phi^{+}}\,,\\
\ket{\Phi^{-}}&=\openone\otimes\sigma^z\ket{\Phi^{+}}\,,\\ 
\ket{\Psi^{+}}&=\openone\otimes\sigma^x\ket{\Phi^{+}}\,,\\ 
\ket{\Psi^{-}}&=\openone\otimes\sigma^x\sigma^z\ket{\Phi^{+}}=\openone\otimes\sigma^y\ket{\Phi^{+}}\,.
\end{align}
These relations allow us to write Eqs.~\eqref{eq:povm1}-\eqref{eq:povm4}, which has the structure of  Heisenberg-Weyl (HW) group covariant SIC POVMs. Upon fixing $\xi$ to be a fiducial state of the  HW group SIC POVM, as we do above, we recover the tetrahedron measurement.

\section{Generalization of RDM measurement scheme to qudits}
\label{sec:sic_qudits}

In this appendix, we generalize the scheme introduced in Sec.~\ref{sec:RDM} to $n$-qudit systems, i.e., collections of $\DD$-level spin systems.  We prove that our scheme is informationally completely and can be implemented using the generalized Bell basis.  Consider the Heisenberg-Weyl (HW) shift and phase (clock) operators, 
\begin{align}
 X\ket{\dd} = \ket{\dd\oplus 1}\,,\quad Z\ket{\dd} = e^{\frac{2\pi i \dd}{\DD}}\ket{\dd} \,,
\end{align}
where $\dd =0,\ldots, \DD-1$ and $\oplus$ is addition modulo $\DD$ ($\ominus$ for modular subtraction).  They obey the HW commutation relation $ XZ =  e^{\frac{2\pi i}{\DD}}ZX$.
Similarly to the Pauli operators on a qubit, the HW operators generate a complete operator basis acting on the $\DD$-dimensional Hilbert space. In particular, a density matrix $\rho$ of a qubit can be expressed as 
\begin{align}
 \rho = \frac1{\DD}\bigg(\openone+\sum_{f\land g\neq0}^{\DD-1} \rho^{f\ssp g} X^f Z^g\bigg)\,,
\end{align} 
where $ \rho^{f\ssp g}=\tr(Z^{-g} X^{-f} \rho)$ are the expansion coefficients. Similarly, given an $n$-qudit density matrix $\rho$, a $k$-qudit RDM can be written as
\begin{align}
  &\tr_{\substack{\neq j_1,\ldots, j_k}}\big( \rho\big) \nonumber\\&= \frac{1}{\DD^k}\sum_{f_1,g_1\ldots,f_k,g_k = 0}^{\DD-1} \rho^{f_1 g_1,\ldots,f_k g_k }_{j_1,\ldots, j_k}\Motimes_{i=1}^k X^{f_i}_{j_i}Z^{g_i}_{j_i}\,,
\end{align}
 where $X^{f_i}_{j_i}Z^{g_i}_{j_i}$ are the HW operators on the $j_i$-th qudit, and 
\begin{align}
\rho^{f_1 g_1,\ldots,f_k g_k}_{j_1,\ldots, j_k} = \tr\bigg(\rho\,\Motimes_{i=1}^k X^{f_i}_{j_i}Z^{g_i}_{j_i}\bigg)\,,
\end{align}
By measuring the above correlation functions for all $f_1,h_1\ldots,f_k,h_k$ and all $j_1,\ldots, j_k$  we can determine all $k$-qudit RDMs. 

While the HW operators on a qudit are not commute or anticommute as the Pauli operators on a qubit, nevertheless, we can measure all the required correlation functions with a single quantum circuit, as we describe below. Consider the generalization of the Bell-basis to qudits, 
\begin{align}\label{eq:generalized_bell_basis}
\ket{\Phi_{h\ell}}& = \openone\otimes X^h Z^\ell\,\ket{\Phi_{00}} = Z^\ell X^{-h} \otimes\openone\, \ket{\Phi_{00}}\,,
\end{align}
where
\begin{align}
\ket{\Phi_{00}} = \frac{1}{\sqrt \DD}\: \sum_{\dd=0}^{\DD-1}\: \ket{\dd}\otimes \ket{\dd}\,.
\end{align}
Following the HW commutation relation we obtain,
\begin{align}
\braket{\Phi_{h\ell}}{\Phi_{h'\ell'}} &= \bra{\Phi_{00}}\openone\otimes ( X^{h}Z^{\ell})^\dagger X^{h'} Z^{\ell'}\ket{\Phi_{00}}\\
&=\delta_{\ell,\ell'}\, \delta_{h,h'}\,,
\end{align}
where the $\delta_{\ell,\ell'}=1$ if $\ell\ominus\ell'=0$ and $\delta_{\ell,\ell'}=0$ otherwise (and similarly for  $\delta_{h,h'}$).
Therefore, the set $\{\ket{\Phi_{h\ell}}: h, \ell =0,\ldots, \DD-1\}$ form an orthonormal basis for two qudits.  It is the eigenbasis of the operators $X^f Z^g\otimes X^{f} Z^{-g}$ for $f, g =0,\ldots, \DD-1$,
\begin{align}
  X^f Z^g\otimes X^{f} Z^{-g} \,\ket{\Phi_{h\ell}} = e^{2\pi i\,(gh - f\ell)/\DD}\, \ket{\Phi_{h\ell}}\,.
\end{align}
Hence, measuring these mutually-commuting operators in the generalized Bell basis~\eqref{eq:generalized_bell_basis} reveals their values simultaneously.  

Similar to the qubit case, one can estimate elements in all $k$-qudit RDMs with error $\epsilon$ by repeating the following steps for $(\DD+1)^k / \epsilon^2$ times:
\begin{enumerate}
\item{To each system qudit (labeled by $j$) attach an ancillary qudit (labelled by $j'$) in a known state $\xi$, so that the total system-ancilla state is $\rho\otimes\xi^{\otimes n}$.}
\item{Measure each pair of qudits $(j,j')$ in the common eigenbasis of $ X^f Z^g\otimes X^{f} Z^{-g}$, i.e., the generalized Bell basis.}
\end{enumerate}
By choosing a qudit ancilla state $\xi$ such that
\begin{align}\label{eq:ancilla_state_qudit}
\tr\Big(X^{f_i}_{j'_i} Z^{-g_i}_{j'_i}\xi\Big)=\frac1{\sqrt{\DD+1}}\,,
\end{align}
we obtain  
\begin{align}
&\rho^{f_1 g_1,\ldots,f_k g_k}_{j_1,\ldots, j_k} =\frac{\tr\Big(\rho\otimes\xi^{\otimes n} \prod_{i=1}^k X^{f_i}_{j_i} Z^{g_i}_{j_i}\otimes  X^{f_i}_{j'_i} Z^{-g_i}_{j'_i} \Big)}{\sqrt{\DD+1}^{\,k}}\,. 
\end{align}
Such states $\xi$ are known in various dimensions and believed to exist in every dimension. For example, analytic solutions are known for dimensions 2-24, 28, 30, 31, 35, 37, 39, 43, 48, 124, and numerical solutions are known in all dimensions up to and including 151~\cite{fuchs_sic_2017}.  The last equation implies that to obtain a fixed standard-deviation error $\epsilon$ in the measured quantities we must run the experiment $(\DD+1)^k\ssp / \epsilon^2$ times, independently of $n$.  In Apps.~\ref{sec:sic_qubit}, we show that this measurement scheme implements a SIC POVM on a qudit. 

For example, for a qutrit (three-dimensional system) the pure state $\psi=\frac1{\sqrt2}(0,1,-1)$ is known to be a fiducial state of the HW group covariant SIC POVM~\cite{Dang2013}.  For the proposed protocol to success we only require that $\big\lvert\tr(X^{f_i}_{j'_i} Z^{-g_i}_{j'_i}\xi)\big\rvert\geq \delta$.  The parameter $\delta$ determines the number of total measurement rounds in the protocol.  Therefore, for dimensions where the HW group covariant SIC POVMs are unknown or if we are only required to estimate specific correlation functions we can replace the condition~\eqref{eq:ancilla_state_qudit} with a less stringent one.

\begin{theorem}\label{thm:qudit_info_comp}
If $\tr(X^{f} Z^{-g}\xi)\neq0$ for $f,g=0,\ldots,\DD-1$, then in the limit of infinite number measurement repetitions the above procedure is informationally complete for $\rho$.
\end{theorem}

\begin{proof}
Since the generalized Bell basis is a common eigenbasis of $ X^f Z^g\otimes X^{f} Z^{-g}$,  for infinite number measurement repetitions, for any $k\in[1,n]$,   $f,g=0,\ldots,\DD-1$ and $j_i\in[1,n]$ ($j_i\neq j_{i'}$ for $i\neq i'$, $i,i'=1,\ldots,k$), using the above procedure we can calculate the $2k$-body correlation function $\tr\big(X^{f_1}_{j_1} Z^{g_1}_{j_1}\otimes  X^{f_1}_{j'_1} Z^{-g_1}_{j'_1} \cdots X^{f_k}_{j_k} Z^{g_k}_{j_k}\otimes  X^{f_k}_{j'_k} Z^{-g_k}_{j'_k}\rho\otimes\xi^{\otimes n}\big)$ which equals to $\tr\big(X^{f_1}_{j_1} Z^{g_1}_{j_1} \cdots X^{f_k}_{j_k} Z^{g_k}_{j_k}\rho\big)\prod_{i=1}^k\tr\big(X^{f_i}_{j'_i} Z^{-g_i}_{j'_i}\xi\big)$.  Since we assumed $\xi$ is known and $\tr\big(X^{f_i}_{j'_i} Z^{-g_i}_{j'_i}\xi\big)\neq0$  for $f_i,g_i=0,\ldots,\DD-1$, we can write
\begin{align}
&\rho^{f_1,h_1\ldots,f_k,h_k}_{j_1,\ldots, j_k} =\tr\big(X^{f_1}_{j_1} Z^{g_1}_{j_1} \cdots X^{f_k}_{j_k} Z^{g_k}_{j_k}\rho\big)\\[2pt]&
=\frac{\tr\big(X^{f_1}_{j_1} Z^{g_1}_{j_1}\otimes  X^{f_1}_{j'_1} Z^{-g_1}_{j'_1} \cdots X^{f_k}_{j_k} Z^{g_k}_{j_k}\otimes  X^{f_k}_{j'_k} Z^{-g_k}_{j'_k}\rho\otimes\xi^{\otimes n}\big)}{\prod_{i=1}^k\tr\big(X^{f_i}_{j'_i} Z^{-g_i}_{j'_i}\xi\big)}\,.
\end{align}
This completes the proof.
\end{proof}

The SIC POVM on a qudit, a $\DD$-dimensional Hilbert space, is a collection of $\DD^2$ POVM elements $E_i=\frac1{\DD}\ket{\psi_i}\bra{\psi_i}$, $i=1,\ldots,\DD^2$, such that 
\begin{align}
\vert\braket{\psi_i}{\psi_j}\vert^2=\frac{\DD\delta_{i,j}+1}{\DD+1}\,. 
\end{align}
We now show that, following Neumark's theorem it can be implemented  by attaching an ancilla qudit is a state $\xi$ (to be determined later) and measuring the two qudits in a generalized Bell-basis.

Consider the following generalization of the Bell-basis to qudits. We define, 
\begin{align}
\ket{\Phi_{h\ell}} = \openone\otimes X^h Z^\ell\; \bigg(\frac{1}{\sqrt \DD} \sum_{\dd=0}^{\DD-1}\, \ket{d}\otimes \ket{d}\bigg)=X^h Z^\ell \otimes\openone \; \bigg(\frac{1}{\sqrt \DD} \sum_{d=0}^{\DD-1}\, \ket{d}\otimes \ket{d}\bigg)\,,
\end{align}
where $X$ and $Z$ are the HW shift and phase (clock) operators, respectively,
\begin{align}
 X\ket{\dd} = \ket{\dd\oplus 1}\,,\quad Z\ket{\dd} = e^{i2\pi \dd/\DD}\ket{\dd} \,,
\end{align}
where $\dd =0,\ldots, \DD-1$ and $\oplus$ is addition modulo $d$. Following the HW commutation relation,
\begin{align}
 XZ =  e^{i2\pi /\DD}ZX \,,
\end{align}
we obtain,
\begin{align}
\braket{\Phi_{h\ell}}{\Phi_{h'\ell'}} &= \bra{\Phi_{00}}\openone\otimes ( X^{h}Z^{\ell})^\dagger X^{h'} Z^{\ell'}\ket{\Phi_{00}}=\bra{\Phi_{00}}\openone\otimes Z^{\DD-\ell} X^{\DD-h} X^{h'} Z^{\ell'}\ket{\Phi_{00}}\\&=e^{i2\pi (h-h')\ell'/\DD}\bra{\Phi_{00}}\openone\otimes Z^{\ell'-\ell} X^{h-h'}\ket{\Phi_{00}}=\delta_{\ell,\ell'}\,\delta_{h,h'}\,,
\end{align}
where the $\delta_{\ell,\ell'}=1$ if $\ell\ominus\ell'=0$ and $\delta_{\ell,\ell'}=0$ otherwise (and similarly for  $\delta_{h,h'}$).
therefore, the set $\{\ket{\Phi_{h\ell}}: h=0,\ldots, \DD-1; \ell=0,\ldots, \DD-1\}$ is a set of orthonormal basis (generalized Bell basis) on two qudits.

Following the same calculation done for two-qubit Bell measurement, it is straight forward to check that measuring the generalized Bell basis basis on a two-qudit state $\rho\otimes\xi$, defines an effective measurement  on the first qudit with $d^2$ POVM elements
\begin{align}
E_{\Phi_{h\ell}} &= \frac1{d} X^{h}Z^{\ell}\xi Z^{\DD-\ell}X^{\DD-h}\,,\;\; h,\ell=0,\ldots, \DD-1.
\end{align}
These equations define the HW group covariant SIC POVM on a qudit, when $\xi$ is its fiducial state. Finding fiducial states for HW covariant SIC POVM is an open question, though known results exist in certain dimensions.

\end{document}